\documentclass[journal,draftcls,onecolumn,12pt]{IEEEtran}
\IEEEoverridecommandlockouts
\IEEEoverridecommandlockouts

\usepackage{epstopdf}  
\usepackage{graphicx} 
\usepackage{times} 
\usepackage{amsmath} 
\usepackage{amssymb}  
\usepackage{amsthm}  
\usepackage{bm}  
\usepackage{mathtools}
\usepackage{epstopdf}
\usepackage{cite}
\usepackage[draft]{hyperref}
\usepackage{framed}
\usepackage{multirow}
\usepackage{parskip}
\usepackage{multirow}
\usepackage{etoolbox}
\usepackage{mathptmx}
\usepackage{calrsfs}
\usepackage{lipsum}
\usepackage{subcaption}
\usepackage{relsize}
\usepackage{setspace}

\DeclareMathAlphabet{\pazocal}{OMS}{zplm}{m}{n}

\newcommand{\bigO}{\mathcal{O}}

\newcommand{\sr}{\stackrel}

\newcommand{\rar}{\rightarrow}

\newcommand{\tri}{\sr{\triangle}{=}}

\newcommand{\be}{\begin{equation}}
\newcommand{\ee}{\end{equation}}
\newcommand{\bea}{\begin{eqnarray}}
\newcommand{\eea}{\end{eqnarray}}
\newcommand{\bes}{\begin{eqnarray*}}
\newcommand{\ees}{\end{eqnarray*}}
\newcommand{\bce}{\begin{center}}
\newcommand{\ece}{\end{center}}
\newcommand{\beae}{\begin{IEEEeqnarray}{rCl}}
\newcommand{\beael}{\begin{IEEEeqnarray}{lll}}
\newcommand{\eeae}{\end{IEEEeqnarray}}
\newcommand{\nms}{\IEEEeqnarraynumspace}
\def\VR{\kern-\arraycolsep\strut\vrule &\kern-\arraycolsep}
\def\vr{\kern-\arraycolsep & \kern-\arraycolsep}

\DeclareMathOperator*{\argmin}{arg\,min}

\newcommand{\ben}{\begin{enumerate}}
\newcommand{\een}{\end{enumerate}}

\DeclareRobustCommand{\bigO}{%
  \text{\usefont{OMS}{cmsy}{m}{n}O}%
}

\newtheorem{theorem}{Theorem}
\newtheorem{proposition}{Proposition}

\newtheorem{remark}{Remark}
\newtheorem{corollary}{Corollary}

\newtheorem{definition}{Definition}
\newtheorem{lemma}{Lemma}

\begin{document}

\title{Finite Blocklength Analysis of Multiple Access Channels with/without Cooperation}


%
\author{Christos K. Kourtellaris, Constantinos Psomas, \IEEEmembership{Senior Member, IEEE}, and Ioannis Krikidis, \IEEEmembership{Fellow, IEEE}
	
\thanks{C. K. Kourtellaris, C. Psomas, and I. Krikidis are with the Department of Electrical and Computer Engineering, University of Cyprus, Cyprus (e-mail: \{kourtellaris.christos, psomas, krikidis\}@ucy.ac.cy).
	
Parts of this work were presented at the IEEE International Conference on Communications, Shanghai, China, May 2019 \cite{kourt019icc1,kourt019icc2}.}}


\maketitle

\begin{abstract} 
Motivated by the demand of reliable and low latency communications, we employ tools from information theory, stochastic processes and queueing theory, in order to provide a comprehensive framework regarding the analysis of a Time Division Multiple Access (TDMA) network with bursty traffic, in the finite blocklength regime. Specifically, we re-examine the stability conditions of a non-cooperative TDMA multiple access channel, evaluate the optimal throughput, and identify the optimal trade-off between data packet size and latency. The evaluation is performed both numerically and via the proposed approximations that result in closed form expressions. Then, we examine the stability conditions and the performance of the  Multiple Access Relay Channel with TDMA scheduling, subject to finite blocklength constraints, by applying a cognitive cooperation protocol that assumes relaying is enabled when sources are idle. Finally, we propose the novel Batch-And-Forward (BAF) strategy, that can significantly enhance the performance of cooperative networks in the finite blocklength regime, as well as reduce the requirement in metadata. The BAF strategy is quite versatile, thus, it  can be embedded in existing  cooperative protocols, without imposing additional complexity on the overall scheme. 
\end{abstract}

\begin{IEEEkeywords}
Finite blocklength analysis, network stability, cognitive cooperation, multiple access relay channels, TDMA, bursty traffic model.
\end{IEEEkeywords}

\IEEEpeerreviewmaketitle

\section{Introduction}


Information theory paved the way for the development of communication theory, and evolved over the years to include a wide range of communication applications, such as, compression, coding, and statistics. However, it has fallen short of leaving its  distinct mark in the field of communication networks \cite{ephremides1998, popo2016}, and this confinement is mainly attributed to the asymptotic nature of Information theory. Shannon's definition of channel's capacity requires infinitely large blocklength in order to guarantee arbitrary small probability of error, for all rates below the channel's capacity. Thus,
classical information measures cannot handle realistic scenarios where 
the blocklength is finite. Finite blocklength is also inextricably linked with the requirement  of Ultra-Reliable Low-Latency Communications (URLLC) \cite{itu2015}, which emerged to support a vast family of applications that require the simultaneous consideration of latency and reliability criteria, and is a key factor for many vertical markets, including, autonomous vehicles, remote healthcare and mission critical communications. The majority of these applications will be supported by current and future wireless communication networks.

Fortunately, recent results  \cite{Polyanskiy2010,tan2015,poly2011} provide valuable tools regarding the analysis of communication networks in the finite blocklength regime. These works, among other results provide attractive approximation for the finite blocklength rate $R^*(n,\epsilon)$, at fixed blocklength $n$, and fixed probability of error $\epsilon$.  These results were  applied to address the requirement of low latency from various perspectives, such as, the characterization of finite blocklength rates for various channels \cite{yang2014, yang16},  the  performance evaluation of short length codes \cite{Wonterghem2018}, and the performance analysis of communication protocols \cite{devassy}. On topics related to cooperation in multiple access channels, though there is an extensive literature that spans from the performance analysis \cite{simeone2007} to protocol design  \cite{sadek2007,krikidis2009a}, and from relay selection\cite{nomikos2016} to full-duplex cooperative relaying\cite{pappas2015}, the vast majority of the existing literature regards asymptotic, in terms of blocklength, analysis. Thus, though these techniques can be employed in the context of finite blocklength, they do not necessarily  perform in an optimal manner. 

The purpose of this work is twofold. First, to provide a comprehensive framework regarding the performance analysis of Time Division Multiple Access (TDMA) channels, in the finite blocklength regime. Towards this direction, we apply tools from information theory, stochastic processes and queueing theory, to study the performance of the network, in terms of stability  and optimal throughput, in the finite blocklength regime. Second, to propose schemes that utilize the potentials of finite blocklength analysis in order to overcome  possible limitations of existing schemes, and enhance the performance of the network.

We begin our analysis by revisiting the stability of the non cooperative Time Division Multiple Access (TDMA) network, subject to finite blocklength constraints, and provide expressions for the optimal throughput of the overall network. The analytical evaluation of the throughput involves the Additive White Gaussian Noise (AWGN) Q-function, and since it cannot be integrated in closed form, we provide  approximations of the AWGN Q-function in order (i) to evaluate closed form expressions for the throughput  and (ii) to identify the trade-off between  the size of the data, $k$, and the channel's blocklength, $n$. Subsequently, we extend the results regarding the stability conditions and the optimal throughput to the case of Multiple Access Relay Channel (MARC) with TDMA scheduling. The selected cognitive cooperation protocol   \cite{simeone2007}, which is based on the underlay cognitive radio concept, assumes relaying is enabled when sources are silent (idle). Although the cognitive cooperation protocol may improve the performance of the network, this improvement is disproportionate to the additional complexity and resources that it entails. The reason for the insufficient performance is that existing cooperative protocols are not designed to perform optimally in the finite blocklength regime. Towards this direction, we propose the novel Batch-And-Forward (BAF) strategy that can significantly enhance the performance of  networks that employ short codes. In the BAF strategy, each terminal is allowed to batch $L$ data packets of length $k$, into a single codeword of fixed length $n$. By employing tools and results from batch queue theory, we provide the stability conditions and identify expressions for the optimal throughput of the overall system. Then, by optimizing over the batching size, $L$, we can significantly enhance the performance of the overall network.

The paper makes the following contributions:\begin{itemize}
\item[i)] We characterize the stability region of the TDMA network subject to finite blocklength constraints. We investigate the concavity properties of the throughput, and evaluate the optimal throughput and the optimal trade-off between data packet size and latency. The evaluation is performed both numerically and via the proposed approximations that result in closed form expressions.
\item[ii)] We characterize the stability region and the optimal throughput  of the MARC-TDMA network subject to finite blocklength constraints, for a particular cognitive cooperation protocol. 
\item[iii)] We propose the BAF strategy which can  improve the performance of the network, in the finite blocklength regime.  We embed this strategy in the discussed cognitive cooperation protocol, where we identify the stability conditions and the expression of the optimal throughput. Then, we show via numerical evaluation that the overall performance is significantly enhanced. Although the performance is evaluated for a particular cooperative protocol, the proposed strategy is quite versatile, thus, it can be embedded in the majority of existing cooperative techniques, without imposing additional complexity.
\end{itemize}

The remainder of this paper is organized as follows. In Section \ref{sec:prelim}, we briefly review the recent results in finite blocklength analysis. In Section \ref{sec:netstab}, we describe the system model and the underlying assumptions, prove the stability conditions for the overall queueing system in the finite blocklength regime, and evaluate the overall  throughput and the optimal trade-off between data length and channel's blocklength. In Section  \ref{sec:coop}, we examine the cooperation in the finite blocklength regime, and in Section \ref{sec:albaf}, we discuss the proposed BAF strategy  and provide numerical evaluation of its performance.

\section{Preliminaries on finite blocklength analysis}\label{sec:prelim}
Let $X$ denote the channel input symbol, $Y$ the channel output symbol and $P_{Y|X}(y|x)$ the conditional distribution of the output given the input. Given a memoryless channel characterized by a conditional distribution $P_{Y|X}(y|x)$, its capacity is given by Shannon's celebrated single letter expression
 \bea
C&=&\max_{p_X(x)}I(X;Y)=\max_{p_X(x)}E\left[i(x;y)\right], 
\eea
where  $E$ is the expectation with respect to the joint distribution $p_{X,Y}(x,y)$, $I(X;Y)$ is the mutual information between the random variable $X$ and the random variable $Y$,
and $i(x;y)\tri\left\{\log\frac{P_{Y|X}(y|x)}{P_Y(y)}\right\}$ is the information density.

Shannon's capacity, has a natural operational definition that associates the rate of information and the reliability, that is, the highest coding rate, in which there exist an encoder-decoder pair that achieve arbitrary small probability of error. The error probability itself is shown to vanish asymptotically with the length of the code, as long as the transmission rate is below capacity. By denoting the optimal rate for fixed blocklength $n$ as $R^*(n,\epsilon)$, and block error probability as $\epsilon$, Shannon's capacity may be redefined as follow.
\begin{equation}
C=\lim_{n\rar\infty}\lim_{\epsilon\rar 0} R^*(n,\epsilon).
\end{equation}
Shannon's capacity has a tremendous theoretical value, however, the prerequisite of infinite length codes severely limits its practical usability. This limitation becomes even more critical for communication applications where low latency is imperative. The above challenge can be addressed via the optimal fixed blocklength rate, $R^*(n,\epsilon)$, which eliminates the necessity of infinitely large codes imposed directly by the definition of  capacity. While, in general, $R^*(n,\epsilon)$ is an NP-hard problem \cite{costa2010,popo2016}, the recent work of Polyanskiy, Poor and Vedru \cite{Polyanskiy2010}, among others, refines Strassen's normal approximation of  $R^*(n,\epsilon)$ \cite{Strassen}, and provides an attractive expression for it. In particular, they proved that for a class of channel models with positive capacity, $C$, $R^*(n,\epsilon)$ is given by
\begin{equation}
R^*(n,\epsilon)= C-\sqrt{\dfrac{V}{n}}Q^{-1}(\epsilon)+\bigO\left(\dfrac{\log{n}}{n}\right),
\end{equation}
where $C$ is the ergodic capacity, $V$ is the channel's dispersion, which is by definition the minimum variance of information density over all capacity achieving input distributions \cite{Polyanskiy2010},  $Q^{-1}(\cdot)$ is the inverse of the Gaussian Q-function and $\bigO
({\log{n}}/{n})$ comprises of the higher order terms. For the AWGN channel, the channel's capacity and dispersion are given by 
\bea
C&=&\dfrac{1}{2}\log_2(1+SNR),
\eea
and
\bea
V&=&\dfrac{SNR}{2}\dfrac{SNR+2}{(SNR+1)^2}(\log_2e)^2,
\eea
respectively, where $SNR$ denotes the signal to noise ratio, while the finite blocklength rate subject to equal-power constraint is approximated by
\bea
R^*(n,\epsilon)\approx  C-\sqrt{\dfrac{V}{n}}Q^{-1}(\epsilon). \label{eq:finiterate}
\eea
Substituting $R^*(n,\epsilon)=\frac{k}{n}$, where $k$ denotes the size of the data packet, and solving with respect to the block error probability $\epsilon$, we obtain
\bea
\epsilon(k,n)\approx Q\left(\dfrac{nC-k}{\sqrt{nV}}\right).
\eea
The probability of successful transmission for a code of blocklength $n$, $P_c(k,n)$, is the cumulative distribution function (cdf) of the normal distribution, and it is expressed as
\bea
P_c(k,n)=1-\epsilon(k,n)\approx\dfrac{1}{\sqrt{2\pi}}\int_{-\infty}^{\frac{nC-k}{\sqrt{nV}}}e^{-\frac{z^2}{2}}d{z}. \label{eq:probsucc}
\eea
The recent work in \cite{tan2015}, refined the approximation given in \eqref{eq:finiterate}, by providing the third order term in the normal approximation for the AWGN channel, that resulted in the following expression
\bea
R^*(n,\epsilon)&\approx&  C-\sqrt{\dfrac{V}{n}}Q^{-1}(\epsilon)+\dfrac{\log_2(n)}{2n}, \label{eq:finiteratenew}
\eea
for the finite blocklength rate, and in the following expression
\bea
P_c(k,n)&\approx&1-Q\left(\dfrac{nC-k+0.5\log_2{n}}{\sqrt{nV}}\right),\label{eq:finiteratenewpc}
\eea
for the successful transmission probability.

\section{Stability for the non coperative scheme on the finite blocklength regime}\label{sec:netstab}
In this section, we characterize the stability region and evaluate the performance of the  TDMA scheme, in the finite blocklength regime.  Moreover, we evaluate the optimal throughput and the trade-off between data size and blocklength, both numerically and via the proposed approximations.

\begin{figure}
\center
  \includegraphics[width=0.7\linewidth]{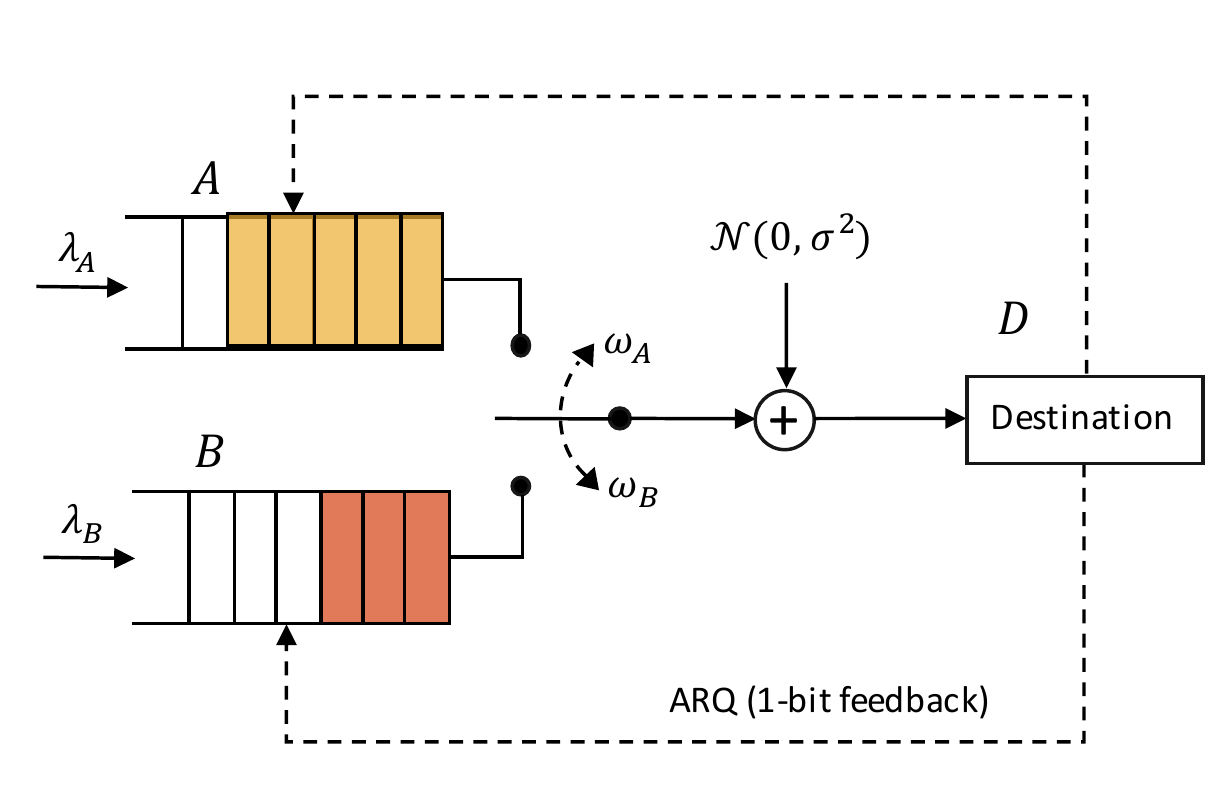}
  \caption{Model of a TDMA network with ACK/NACK feedback.}
  \label{fig:model1}
\end{figure}

\subsection{System Model}

We consider a model with two source terminals, $A$ and $B$, with infinite buffer memories, and a single destination node $D$, as depicted in Fig. \ref{fig:model1}. At each time slot, data packets of length $k_i, \ i\in\{A,B\}$, arrive at the source terminal $i\in\{A,B\}$, according to a Bernoulli distribution with probability $p_i$. The expected value of arrivals at each time slot  is $\lambda_i=p_i, \ \forall  i\in\{A,B\}$. The terminals then encode the data packet  into a codeword of length $n$, and access the channel through a TDMA scheduling with probability $\omega_i$, where $0\leq\omega_i\leq 1$, and $\omega_A+\omega_B=1$ \cite{krikidis2011}. We assume that at each  time slot,  $n$ channel uses  are employed and solely allocated to source terminal $i$, with probability $\omega_i$. The channel is an AWGN channel with zero mean and variance $\sigma^2$. The destination, after receiving and decoding the codeword, sends  Acknowledgement/Negative-Acknowledgement (ACK/NACK) back to the respective source terminal, to inform it about the status of the transmission. In the case of a correct transmission, the respective source terminal discards the data packet from its buffer memory. In the opposite case, the data packet remains in the buffer memory and waits for the next available time slot for retransmission.

The probability of an erroneous transmission for a packet, generated by terminal $i$ at a given time slot, is denoted by $P_{e}(k_i,n)$. The service (departure) process is Bernoulli distributed with probability $q_i=\omega_i(1-P_{e}(k_i,n))$.  Since both the arrivals and departures are Bernoulli distributed, the time of an arrival and the time for a departure to occur, measured in slots, is characterized by a geometric distribution. The system at each terminal $i\in\{A,B\}$ can be described by a  discrete time Markov process with states $\{S_j, j\geq 0\}$, which denote the number of packets in the system. 

\subsection{Stability Conditions and Optimal Throughput}

Our first objective is to study the maximum rate that can be supported by the network. Towards this direction, we prove that network stability is possible, if and only if, the overall rate of the system is less than the throughput.
\begin{theorem}\label{theo:main1}
The TDMA network is stable, if and only if, the following conditions hold
\bea
 {\lambda_i}&<&\omega_i P_{c}(k_i,n),  \ \ \forall i\in\{A,B\}, \  \label{eq:throughputeq0}\\
 {\lambda_A}+ {\lambda_B}&<&\omega_A P_{c}(k_A,n)+\omega_B P_{c}(k_B,n). \label{eq:throughputeqd0}
\eea
\end{theorem}
\begin{proof}
 See Appendix \ref{appendix:theom1}.
\end{proof}

The following corollary is a straightforward consequence of Theorem~\ref{theo:main1}.

\begin{corollary}\label{theo:main1old}
Let $X(k,n)\tri\frac{k}{n}(\lambda_A+\lambda_B)$ denote the rate of the  scheme,  $u(k,n)\tri\frac{k}{n}P_c(k,n)$ denote the overall throughput of the  scheme\cite{Polyanskiy2010}, and $k_A=k_B=k$. Then,

\bea  
X(k,n)\tri(\lambda_A+\lambda_B)\dfrac{k}{n}<\dfrac{k}{n}P_c(k,n)\tri u(k,n). \label{eq:throughputeq}
\eea
\end{corollary}
\begin{proof}
 For the special case where $k_A=k_B=k$, then $P_{c,A}(k_A,n)=P_{c,B}(k_B,n)=P_{c}(k,n)$, thus from \eqref{eq:throughputeq0}, we have
\begin{equation}
 {\lambda_i}<\omega_i P_{c}(k,n),  \ \ \forall i\in\{A,B\}.
 \end{equation}
Since $\omega_A+\omega_B=1$, the stability for the overall scheme consisted of the two terminals $A$ and $B$, is calculated, as follow
\begin{equation}
\lambda_A+\lambda_B<(\omega_A+\omega_B)P_{c}(k,n)=P_{c}(k,n). \label{eq:stab_case1}
\end{equation}

Multiplying both sides of \eqref{eq:stab_case1} with $\frac{k}{n}$, we obtain \eqref{eq:throughputeq}.
\end{proof}

The assumption $k_A=k_B=k$ is imposed to keep the notation clean. However,  the general case of different data packet size can emerge directly by employing the proposed analysis.

Next, we employ Corollary \ref{theo:main1}, to recast the classical problem of maximizing the overall rate of the network by imposing a blocklength (latency) constraint. That is, given a channel and a fixed blocklength  $n$, we ask what is the optimal size of the data packets that maximizes the  rate. 
For the rest of this work, we will consider the case where the size of the data arriving at the two terminals is identical, that is, $k_A=k_B=k$, investigate the impact of the blocklength, $n$, on the throughput, and provide numerical evaluation and closed-form approximations for the throughput.

As proved in Corollary~\ref{theo:main1}, the overall rate that guarantees stability can be arbitrary close to the throughput of the system. Thus, to maximize rate, we need to identify the optimal value of $k$ that maximizes $u(k,n)$. The resulted optimization problem is given by
\bea
u^{*}(k,n)=\max_{k}\dfrac{k}{n}P_c(k,n). \label{optim_pro1}
\eea
Before we proceed to the solution of the above optimization problem, we investigate the convexity properties of the objective function $u(k,n)$. Towards this direction, we state the necessary definition of log-concavity and a lemma which highlights an important property of log-concave functions.

\begin{definition}\label{def:logconc}
A function $f: \mathbb{R}^n \mapsto \mathbb{R}$ is  log-concave if $f(x) > 0 \ \forall \ x$, and $\log{f}$ is concave.\end{definition}

\begin{lemma}\label{lemma:multipli}
Log-concavity is closed under multiplication, that is, if $f$ and $g$ are log-concave, the pointwise product is also log-concave \cite[Section~3.5]{boyd2004}.
\end{lemma}
We now state the theorem regarding the log-concavity of the objective function $u(k,n)$.

\begin{theorem}\label{theo:concavity}
For any fixed $n>1$, $u(k,n)$ is log-concave function of $k$.
\end{theorem}

\begin{proof}
Let $f(k)=\frac{k}{n}$  and $h(k)=P_c(k,n)$. The objective function can be rewritten as $u(k)=f(k)h(k)$. By Definition \ref{def:logconc}, $f(k)$ is log-concave since $f(k)>0$ and  $\log{f(k)}$ is concave. The function $h(k)$ is by definition the cdf of a normal distribution, which is shown to be log-concave \cite[Section~3.5]{boyd2004}.  Since both the functions $f(k)$ and $h(k)$ are log-concave, then by Lemma \ref{lemma:multipli}, the function $u(k)$ is also log-concave.
\end{proof}

By virtue of Theorem \ref{theo:concavity},  $u(k,n)$ is unimodal, that is, there are no local maxima that are non-global ones. This property eliminates the risk for the optimization algorithm getting trapped into a local maxima that is not global. Moreover, log-concavity allows transforming the original optimization problem into a convex optimization problem, that inherits all useful properties and tools of convex optimization.

Unfortunately, no closed form solutions can emerge from the optimization problem \eqref{optim_pro1}, since no explicit expression is known for $P_c(k,n)$. To overcome this problem, we capitalize the properties of the objective function, $u(k,n)$, and provide numerical evaluation of the optimal value of $k$ via exhaustive search. Additionally, we propose first order and second order approximations of $P_c(k,n)$, which are applied in order to evaluate closed form approximations of the optimal data packet size, $k^*$, and the optimal throughput, $u^{*}(k,n)$, with a view to identify the optimal trade-off between the optimal size of the data packet, $k$, and the blocklength $n$. 

\begin{remark}
In our analysis, we do not address the issue of control signals (metadata), which are necessary, inter alia, for the error detecting schemes required for the ACK/NACK protocol. Thus, the results of this work should be interpreted in the light of this consideration. This is translated as a genie aided destination \cite{steger2008,kumar2009}, that can identify possible errors, and requests, or does not request, data retransmission. 
\end{remark}

The optimal solution of the optimization problem \eqref{optim_pro1} can be found via exhaustive search over all possible values of $k\geq 1$. This approach is computationally efficient due to the log-concavity  of $u(k,n)$, which results to a unique global maxima.
The exhaustive search algorithm simply compares the objective function, $u(k,n)$, for successive values of $k$, and terminates the search when $u(k=i+1,n)<u(k=i,n), \ i\in[1,\infty)$. Then, the optimal solution is given by, $k^*=i$. By substituting the value of $k^*$ in \eqref{optim_pro1}, we obtain the value of the throughput.

The analytical evaluation of the throughput involves the AWGN Q-function, and since it cannot be integrated in closed form, 
tight approximations should be employed in order to evaluate closed form expressions for the throughput and for the trade-off between  the size of the data, $k$, and the channel's blocklength, $n$. Despite the significant work on approximations of the Gaussian Q-function (see \cite{karagiannidis2007} and references within), these  
cannot be employed to provide closed form expressions of the throughput, due to their complex structure. Towards this direction, we   propose linear and quadratic approximations on the probability of successful transmission, that result in closed form expressions.

\begin{remark}\label{rem:why}
It has been observed, via numerical evaluation of the  throughput, that the approximation given in \eqref{eq:finiteratenewpc}, though tighter than  \eqref{eq:probsucc} for relatively large blocklength, $ n>10^3$, may produce inconsistent results for very small blocklengths, $n<10^2$ (the approximated rates are greater than channel's capacity). This observation holds especially for small values of SNR ($SNR<1$). Thus, we employ the pessimistic expression \eqref{eq:probsucc} rather than \eqref{eq:finiteratenewpc}. Nevertheless, the proposed methodology and results can be straightforwardly extended to any possible expression of $P_c(k,n)$. 
\end{remark}

\subsection{Q-function Approximations}\label{sec:throueva}
\subsubsection{Linear}
Linear approximations, though not the tightest, are attractive since they provide simple expressions that can be physically interpreted. Recent works on topics related to finite blocklength analysis employ such approximations, for the finite blocklength analysis of the incremental redundancy Hybrid ARQ (HARQ) \cite{makki2014} and for full-duplex and half-duplex relaying for short packet communications \cite{gu2017}. Let, the linear approximation of the probability of successful transmission be denoted by ${\hat{P}_c}(k,n)$, and the resulting approximations of the throughput and of the  data packet size be denoted by  ${\hat u}(k,n)$ and ${\hat k}$, respectively.

\begin{figure}
\center
\begin{minipage}{1\textwidth}
\center
\hspace{-0.8cm}
  \includegraphics[width=0.8\linewidth]{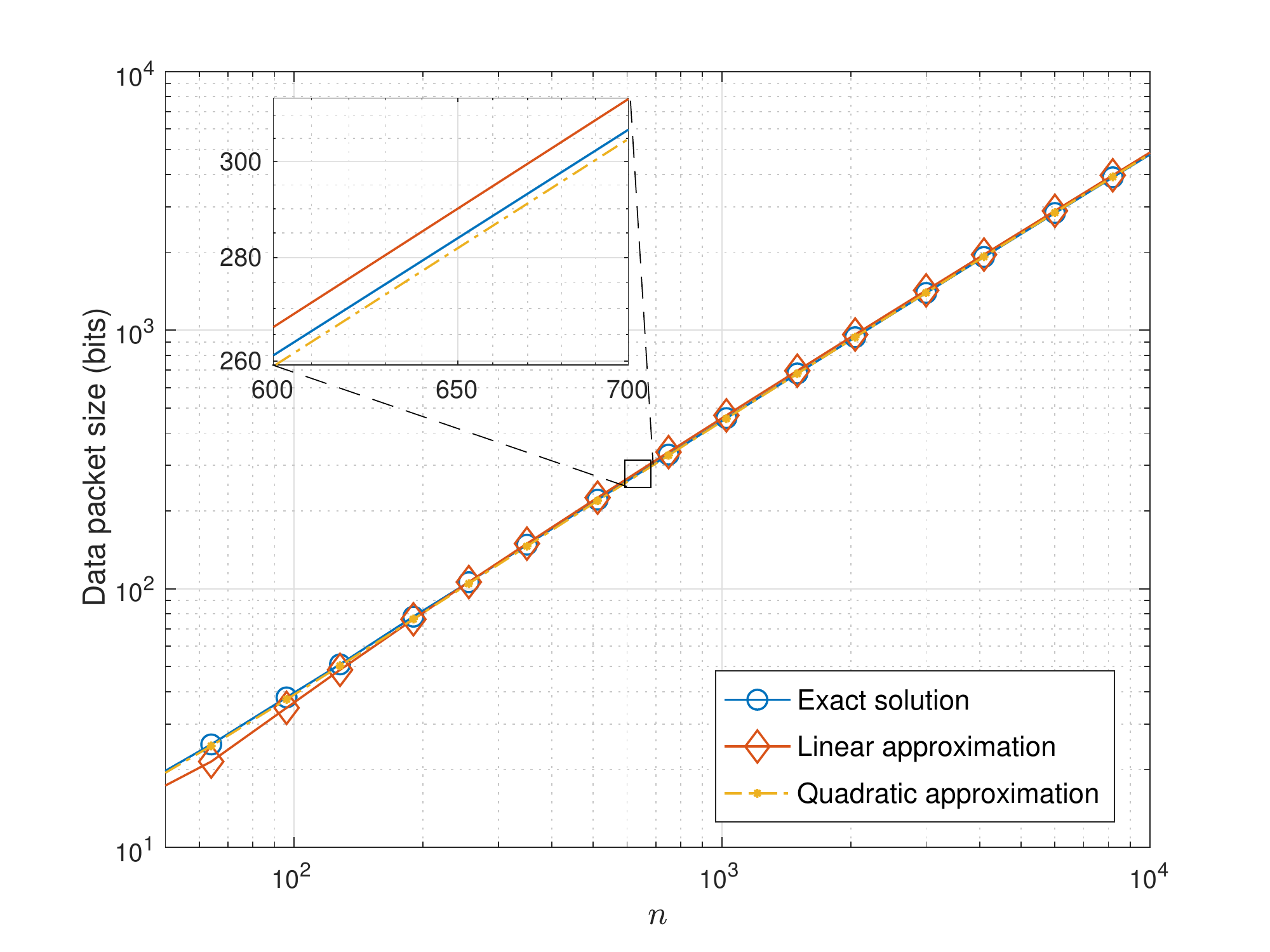}
  \caption{Optimal size of data packets as a function of channel's blocklength, $n$, and comparison with the 
expressions resulted from the 
linear and quadratic approximation of $P_c(k,n)$, for $SNR=1$.}
  \label{fig:compklq}
\end{minipage}%
\hspace{0.4cm}
\begin{minipage}{1\textwidth}
\center
\vspace{-0.5cm}
\hspace{-0.5cm}
 \includegraphics[width=0.8\linewidth]{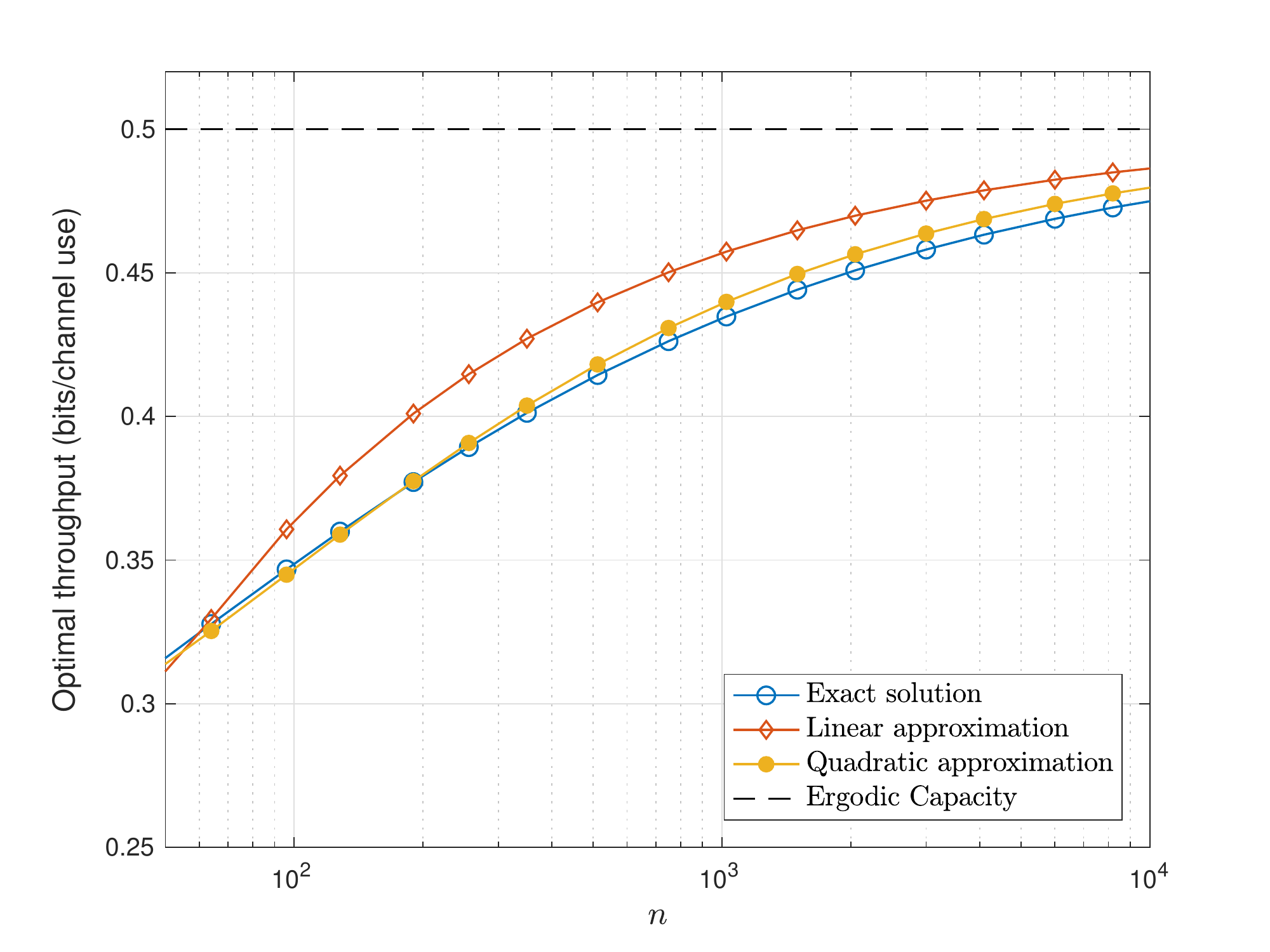}
\caption{ Optimal throughput, $u^{*}(k,n)$, and comparison with the expressions resulted from the 
linear and quadratic approximation of $P_c(k,n)$, for $SNR=1$.}
  \label{fig:boundex}
\end{minipage}
\end{figure}

%

\begin{proposition}\label{prop:lin}
The proposed linear approximation is given by
\begin{equation}
  {\hat{P}_c}(k,n) =
  \begin{cases}
    1 & \text{if $\chi \geq \delta_1$}, \\
    \dfrac{1}{2\delta_1}\chi+\delta_0& \text{if $-\delta_1 \leq \chi < \delta_1$}, \\
    0 & \text{if $\chi <-\delta_1$},
  \end{cases}
\end{equation}
where $\delta_0=0.5$, $\delta_1=1.545$, and 
\beae
\chi&=&\dfrac{nC-k}{\sqrt{nV}}.\label{eq:valueofx}
\eeae
The optimal size of the data packet is given by
\beae
{\hat k}^{*}&=&
  \begin{cases}
    nC-1.545\sqrt{nV} & \text{if $n\geq\dfrac{9{\delta_1}^2V}{C^2}$} \vspace{0.2cm}, \\
    0.5\left(Cn+1.545\sqrt{nV}\right) & \text{if $0< n<\dfrac{9{\delta_1}^2V}{C^2}$}, \label{eq:finaloptsollin}
  \end{cases} 
\eeae
and the optimal value of the throughput is obtained by substituting ${\hat k}^{*}$ in 
\bea
{\hat{u}}^{*}(k,n)=\dfrac{{\hat k}^{*}}{n}{\hat{P}_c}({\hat k}^{*},n).
\eea
\end{proposition}
\begin{proof}
See Appendix \ref{app_prop:lin} .
\end{proof}
The result in Proposition \ref{prop:lin}, and in particular \eqref{eq:finaloptsollin}, provides the optimal trade-off between data size and channel's blocklength.
Note, that since the data size is integer, the optimal solution given in \eqref{eq:finaloptsollin} should be rounded to the nearest integer. Since we are interested in an approximation of the throughput and not its exact calculation, 
the effect of the selected rounding function (i.e., round, ceiling or floor) is negligible. The approximation of the throughput is then obtained by substituting the rounded value of \eqref{eq:finaloptsollin} in \eqref{eq:linappopt}.

Next, we propose a quadratic approximation of $P_c(k,n)$, that, in general, gives tighter results compared to the linear approximation.\bigskip

\subsubsection{Quadratic}
The proposed approximation is quadratic in a defined region of $\chi$ and linear in the rest of the region. Let, the quadratic approximation of the probability of successful transmission be denoted by ${\tilde{P}_c}(k,n)$, and the resulting approximations of the throughput and of the   data packet size be denoted by  ${\tilde u}(k,n)$ and ${\tilde k}$, respectively.

\begin{proposition}\label{prop:quad}
The proposed quadratic approximation is given by
\begin{equation}
  {\tilde{P}_c}(k,n) =
  \begin{cases}
    1 & \text{if $\chi \geq \theta_1$}, \\
    \theta_2\chi(2\theta_1-\chi)+\theta_0& \text{if $0 \leq \chi < \theta_1$}, \\
    \theta_2\chi(2\theta_1+\chi)+\theta_0 & \text{if $-\theta_1 < \chi <0$}, \\
    0 & \text{if $\chi \leq -\theta_1$},
  \end{cases}\label{eq:sec_ord_app}
\end{equation}
where $\theta_0=0.5$, $\theta_1=2.35$, $\theta_2=0.5/{\theta_1}^2$ and $\chi$ is given by \eqref{eq:valueofx}.

The optimal size of the data packet is given by
\begin{equation}
\tilde{k}^{*}=
  \begin{cases}
    \dfrac{2}{3}\left(nC-\theta_1\sqrt{nV}\right)+\theta_3
     & \text{if $n\geq \dfrac{\theta_1^2V}{4C^2}$} \vspace{0.2cm}, \\ 
    \dfrac{1}{3}\left(nC-\theta_1\sqrt{nV}\right) & \text{if $0< n<\dfrac{\theta_1^2V}{4C^2}$}, \label{eq:finaloptsolquad}
  \end{cases} 
 \end{equation}
 where
 \begin{equation}
\theta_3=\dfrac{\sqrt{n}}{3}\left(nC^2-7\theta_1^2V-2\theta_1C\sqrt{nV} \right)^{\frac{1}{2}},
\end{equation}
and the optimal value of the throughput is obtained by
\bea
{\tilde u}^{*}(k,n)=\dfrac{\tilde{k}^{*}}{n}{\tilde{P}_c}(\tilde{k}^{*},n).
\eea
\end{proposition}
\begin{proof}
See Appendix \ref{app_prop:quad} .
\end{proof}

The optimal trade-off between the data packet size and the channel's blocklength, $n$, as well as the comparison with the provided approximations $\hat{k}^{*}$ and $\tilde{k}^{*}$, are depicted in Fig.~\ref{fig:compklq}. While both approximations perform  well, the optimal data packet size emerged from the quadratic approximation, $\tilde{k}^{*}$, is almost identical to $k^*$. The optimal throughput and  the throughput approximations are illustrated in Fig.~\ref{fig:boundex}. Again, the solution emerged from the quadratic approximation approaches very well the numerical evaluation of the optimal throughput.

\begin{remark}
The results of this section can be employed in order to identify the optimal throughput of various schemes, such as, cognitive communication schemes. For example, assume that data of length $k_A$ arrive at the terminal of the primary user (e.g. Terminal A) with rate $\lambda_A$.  By \eqref{eq:throughputeq0} we can determine the minimum $\omega_A$, denoted by $\omega_A^*$, such that \eqref{eq:throughputeq0} holds (assuming it exists). Then, $\omega_B^*=1-\omega_A^*$, while the throughput of the secondary user (e.g. Terminal B) is $\omega_B\frac{k_B}{n}P_c(k_B,n)$. Thus, optimizing over $k_B$, either numerically or via the approximations, we can identify the optimal data length $k_B$ and the optimal throughput both for the secondary user and the overall network.
\end{remark}

\section{Cooperation in the finite blocklength regime}\label{sec:coop}

In this section, we examine a packet-based network cooperation scenario with bursty arrivals at the source terminals. In particular, we consider a  MARC scheduling and evaluate the performance of a cognitive cooperative protocol in the finite blocklength regime. 

The MARC configuration is consisted of two source terminals, $A$ and $B$, a common cognitive relay and a destination, as depicted in Fig. \ref{fig:coope}.  The data packets arrive at the source terminals, $A$ and $B$, according to independent and stationary Bernoulli processes with 
probabilities, $p_A$ and $p_B$, and expected values, $\lambda_A=p_A$ and $\lambda_B=p_B$, respectively. Each of the source terminals has an infinite size buffer memory, denoted by $Q_i, \ i\in\{A,B\}$, respectively, that stores the incoming data packets. The relay  is equipped with two relaying queues, denoted by $Q_{AR}$ and $Q_{BR}$, in which they store the data packets received from the respective source terminals. Let $P_{c,SD}(k,n)$, $P_{c,SR}(k,n)$ and $P_{c,RD}(k,n)$
denote the probability of a successful transmission from
any source terminal $i\in\{A,B\}$ to the destination, from any source terminal $i\in\{A,B\}$ to the relay, and from the relay to the destination, respectively. The probabilities of erroneous transmissions are then defined by $P_{e,SD}(k,n)=1-P_{c,SD}(k,n)$ and $P_{e,SR}(k,n)=1-P_{c,SR}(k,n)$. Moreover, let $u^{CC}(k,n)$ and $u^{CC,*}(k,n)$ denote the throughput and the maximum throughput of the system, respectively, and $X^{CC}(k,n)$ denote the code rate of the overall scheme.
\begin{figure}
\center
  \includegraphics[width=0.7\linewidth]{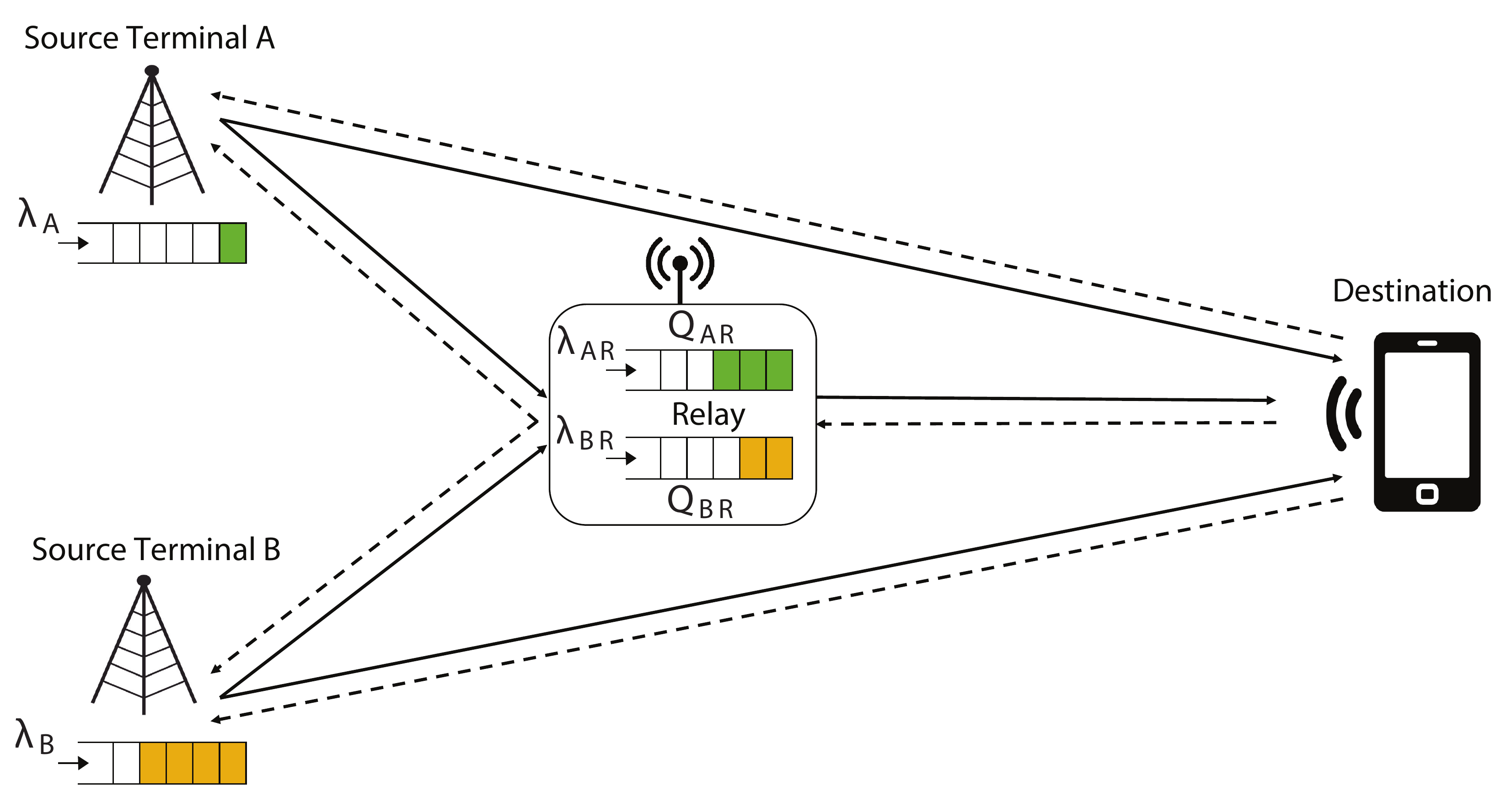}
  \vspace{0.4cm}
  \caption{Model of a MARC-TDMA network. The solid arrows represent the AWGN channels and the dashed arrows the ACK/NACK feedback.}
  \label{fig:coope}
\end{figure}


There is an extensive literature regarding  multiple access protocols in the presence of a cooperating relays \cite{simeone2007,sadek2007,krikidis2009a}. In this work we employ the  cognitive cooperation protocol, defined below.

\begin{definition}\label{def:ccp} 
The Cognitive Cooperation (CC) protocol performs as follows:
\begin{itemize}
\item[i)] Source terminal $i\in\{A,B\}$ encodes the data packet of length $k$ into a codeword of length $n$, and access the channel via a randomized TDMA scheduling with probability $\omega_i, \ i\in\{A,B\}$, and $\omega_A+\omega_B=1$. 
\item[ii)] The codeword is transmitted both to the destination and the relay node. The transmission process is supported by an ACK/NACK mechanism that informs the source terminal and the relay about the transmission status (successful or erroneous).
\item[iii)] If data are not successfully received by either the destination or the relay, the data packet remains in the queue of the source terminal.
\item[iv)] If data are successfully received by the destination, the source terminal removes the data packet from its queue, and the relay  ignores the received packet.
\item[v)] If data are not successfully received by the destination but are successfully received by the relay, the source terminal discards the data packet from its queue, and the relay adds that packet to the respective queue ($Q_{AR}$ or $Q_{BR}$).
\item[vi)] When the source terminal $i\in\{A,B\}$ gains access to the channel but it has no data packets in its queue (queue is idle), the relay encodes a data packet from the respective queue $Q_{iR}, i\in\{A,B\}$, into a codeword of length $n$, and transmits it to the destination.
\end{itemize}
\end{definition}
This protocol, though not ideal in terms of performance, is attractive due to its elegance and simple structure, that allows the interpretation of the results in the context of finite blocklength codes. Nevertheless, the discussed methodology can be applied to more complex protocols.

\begin{figure*}%
\hspace{-0.8cm}
\begin{subfigure}{1\columnwidth}
\center
\centering
\includegraphics[width=0.8\columnwidth]{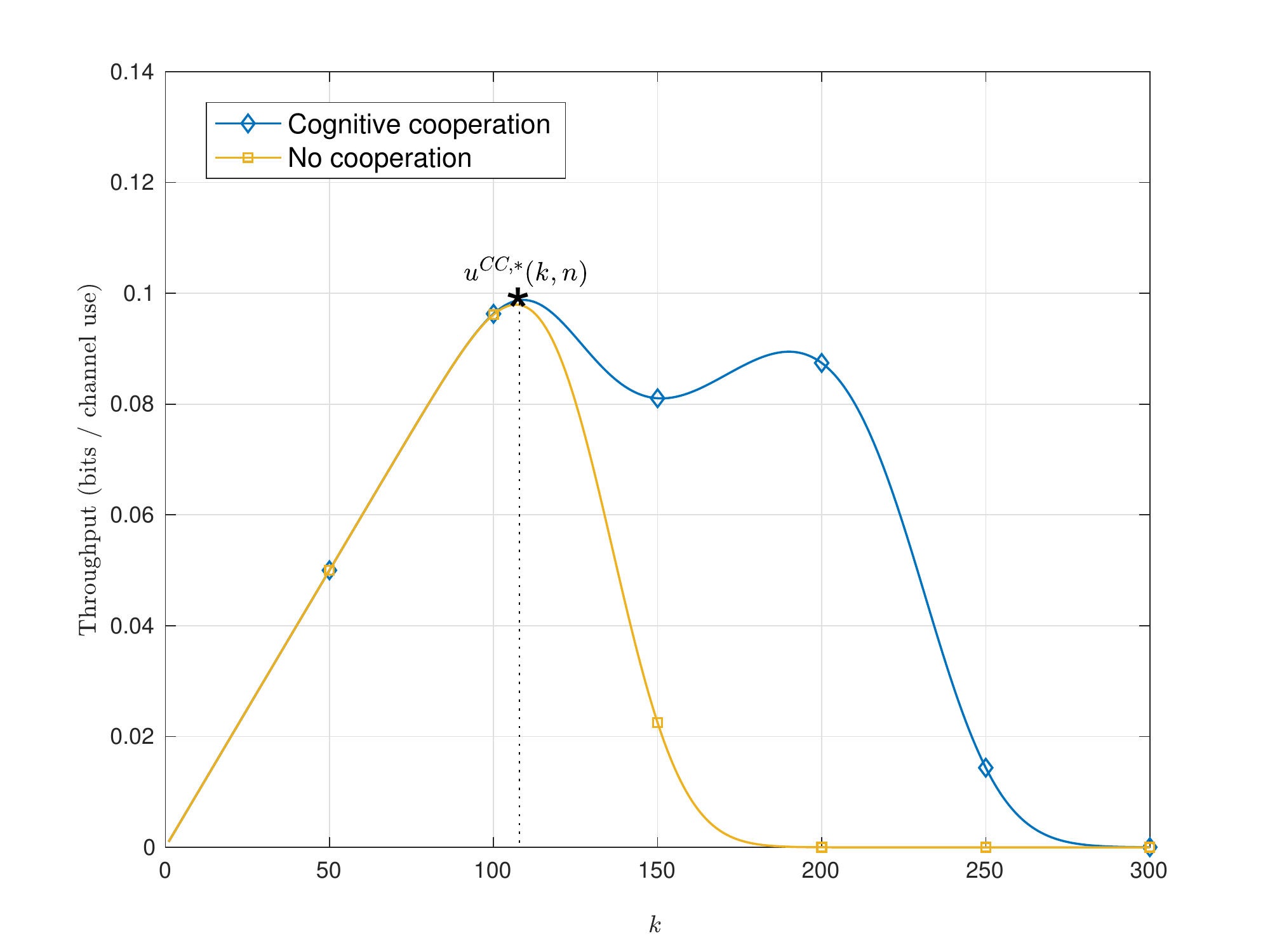}%
\vspace{-0.2cm}
\caption{}%
\label{subfiga}%
\end{subfigure}\hfill%
\begin{subfigure}{1\columnwidth}
\center
\hspace{-0.8cm}
\includegraphics[width=0.8\columnwidth]{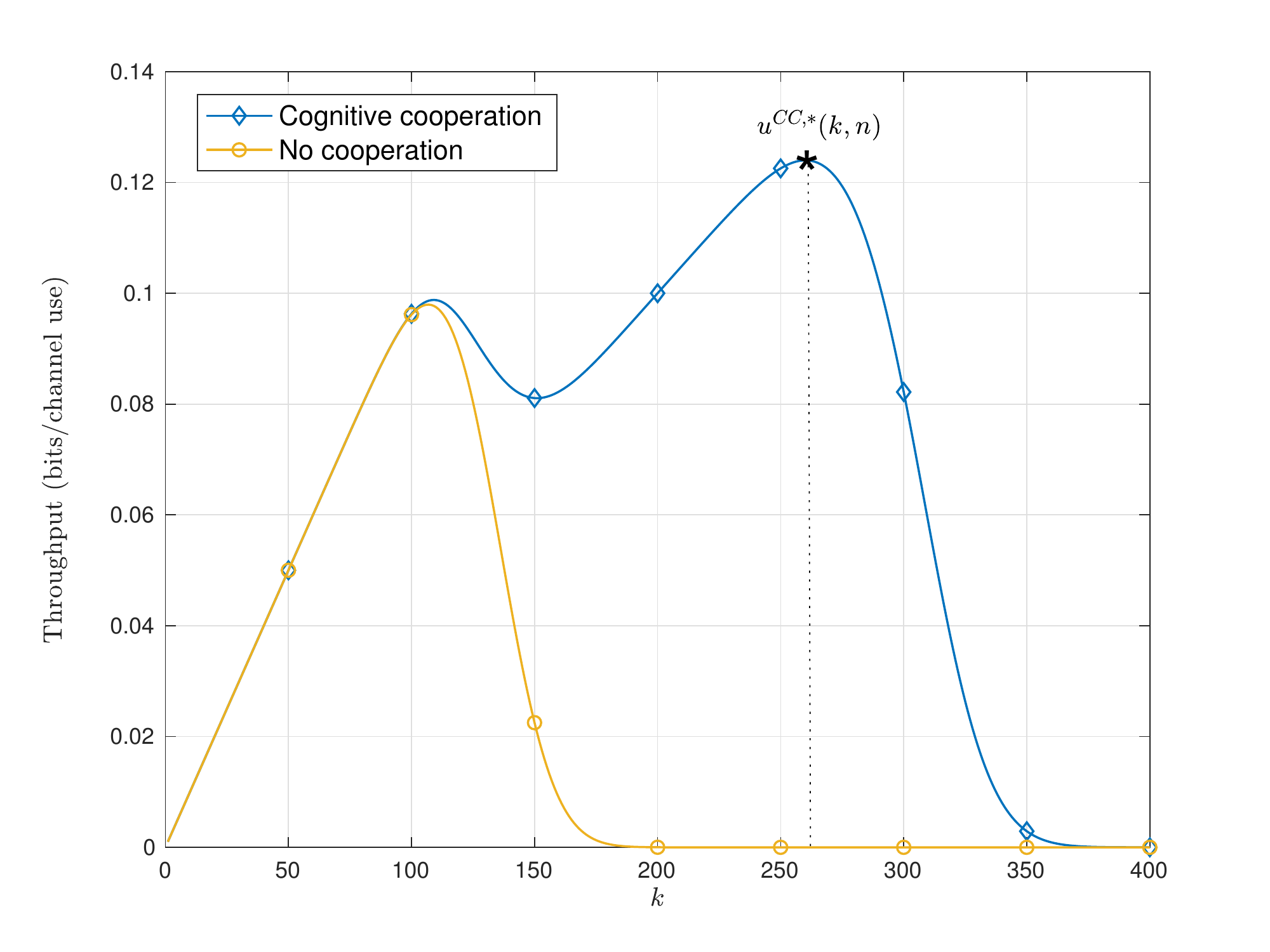}%
\vspace{-0.4cm}
\caption{}%
\label{subfigb}%
\end{subfigure}\hfill%
\caption{Throughput for the non cooperative scheme and the cognitive cooperation scheme,  for fixed blocklength $n=1000$. The channels from the source to the destination, from the source to the relay, and from the relay to the destination, are AWGN with SNR for (a) 0.2, 0.35 and 1, respectively, and for (b) 0.2, 0.5 and 1, respectively.}
\label{figabc}
\end{figure*}

\begin{theorem}\label{theo:coopmarch}
For a fixed blocklength $n$:
\begin{itemize} 
\addtolength{\itemindent}{-0.25cm}
\item[i)] the  stability condition of the CC protocol is given by
\bea
 \frac{\kappa}{n}\left({\lambda_A}+ {\lambda_B}\right)<u^{CC}(k,n), \label{eq_ccsource10}
\eea
\hspace{-0.4cm} where 
\beae
u^{CC}(k,n)&\tri&\frac{k}{n}\dfrac{\left[P_{c,SD}(k,n)+P_{e,SD}(k,n)P_{c,SR}(k,n)\right]P_{c,RD}(k,n) }{\left[P_{c,RD}(k,n)+P_{e,SD}(k,n)P_{c,SR}(k,n)\right]}.\nonumber \\  \label{eq_ccsource10new1}
\eeae
\item[ii)] the code rate of the overall scheme, $X^{CC}(k,n)\tri\frac{k}{n}({\lambda_A}+ {\lambda_B})$, is bounded above by the
 
\hspace{-0.4cm} maximum throughput $u^{CC,*}(k,n)$, that is
\beae
 X^{CC}(k,n)< \max_{k}u^{CC}(k,n) \tri u^{CC,*}(k,n). \label{eq_ccsource10new2a}
\eeae
\end{itemize}
\end{theorem}

\begin{proof}
See Appendix \ref{app_theo:coopmarch}.
\end{proof}

Although the optimization problem defined in \eqref{eq_ccsource10new2a} is not necessarily concave, it can be evaluated via exhaustive search. This does not introduce any additional computational complexity, due to the integer nature of the optimization problem.

The Non Cooperative (NC) protocol (absence of the relay terminal), is a special case of the cooperative protocol, with $P_{c,SR}(k,n)=0$. The overall stability condition is identical to the overall stability condition of source terminals, given
by \eqref{eq_ccsource3}, where $P_{c,SR}(k,n)=0$, which yields the following corollary.

\begin{corollary} 
For a fixed blocklength $n$, the overall stability condition of the non cooperative system  is given by
\beae
  ({\lambda_A}+ {\lambda_B})\frac{k}{n}&<\max_{k}&u^{NC}(k,n)
\tri u^{NC,*}(k,n), \label{eq_ccsource101}
\eeae
where $u^{NC}(k,n)\tri\frac{k}{n}P_{c,SD}(k,n)$.
\end{corollary}

\begin{remark}
The performance, in terms of maximum throughput, of the CC protocol described in Definition \ref{def:ccp}, is not always better than the performance of the NC protocol. For example assume $P_{c,SD}(k,n)>0, \ P_{c,SR}(k,n)>0$ and $P_{c,RD}(k,n)=0$. Then, from \eqref{eq_ccsource10} and \eqref{eq_ccsource101} we have $u^{CC,*}(k,n)=0$ and $u^{NC,*}(k,n)>0$. However, with some minor modifications of the CC protocol, we can guarantee that its performance is always better or at least equal to the performance of the NC protocol.
\end{remark}

The optimal throughput of the non cooperative scheme, $u^{NC,*}(k,n)$, and the optimal throughput of the cognitive cooperation scheme, $u^{CC,*}(k,n)$, for two different channel triplets, is given in Fig. \ref{figabc}. For the selected $SNR$ triplet that is depicted in Fig. \ref{figabc}(a), the increase of the throughput
due to cooperation is negligible, whereas for the $SNR$ triplet  depicted in Fig. \ref{figabc}(b), cooperation increases throughput approximately by $25\%$. However, taking into consideration the commitment of additional resources (relay, buffer memories and channels), the gain in the performance that cognitive cooperation exhibits over the non cooperative scheme, cannot be characterized as satisfactory. Comments for the insufficient performance of the cognitive cooperation protocol are given in the following remark.

\begin{remark}\label{rem:cc}
The expression of the throughput for the overall system,  $u^{CC}(k,n)$, involves the statistical characteristics (SNR) of all available channels. Optimizing throughput with respect to $k$, results in an optimal data packet size $k^*$, that is employed both from the source terminal and the relay. Thus,  $k^*$ emerges as a compromise between the different statistics of those channels. This is directly reflected on the performance of the overall network, since, different channels pack the same amount of data into the codeword of fixed length $n$.   An obvious solution to this problem is to allow the source terminal and the relay to pack different amount of data into the codeword (e.g. source packs $k_S$ bits into the codeword while the relay packs $k_{R}$ bits into the codeword), however, this is highly impractical, since, it introduces significant amount of complexity to the destination.
\end{remark}

\section{Cognitive cooperation via Batch and forward}\label{sec:albaf}
\begin{figure}
\centering
  \includegraphics[width=0.7\linewidth]{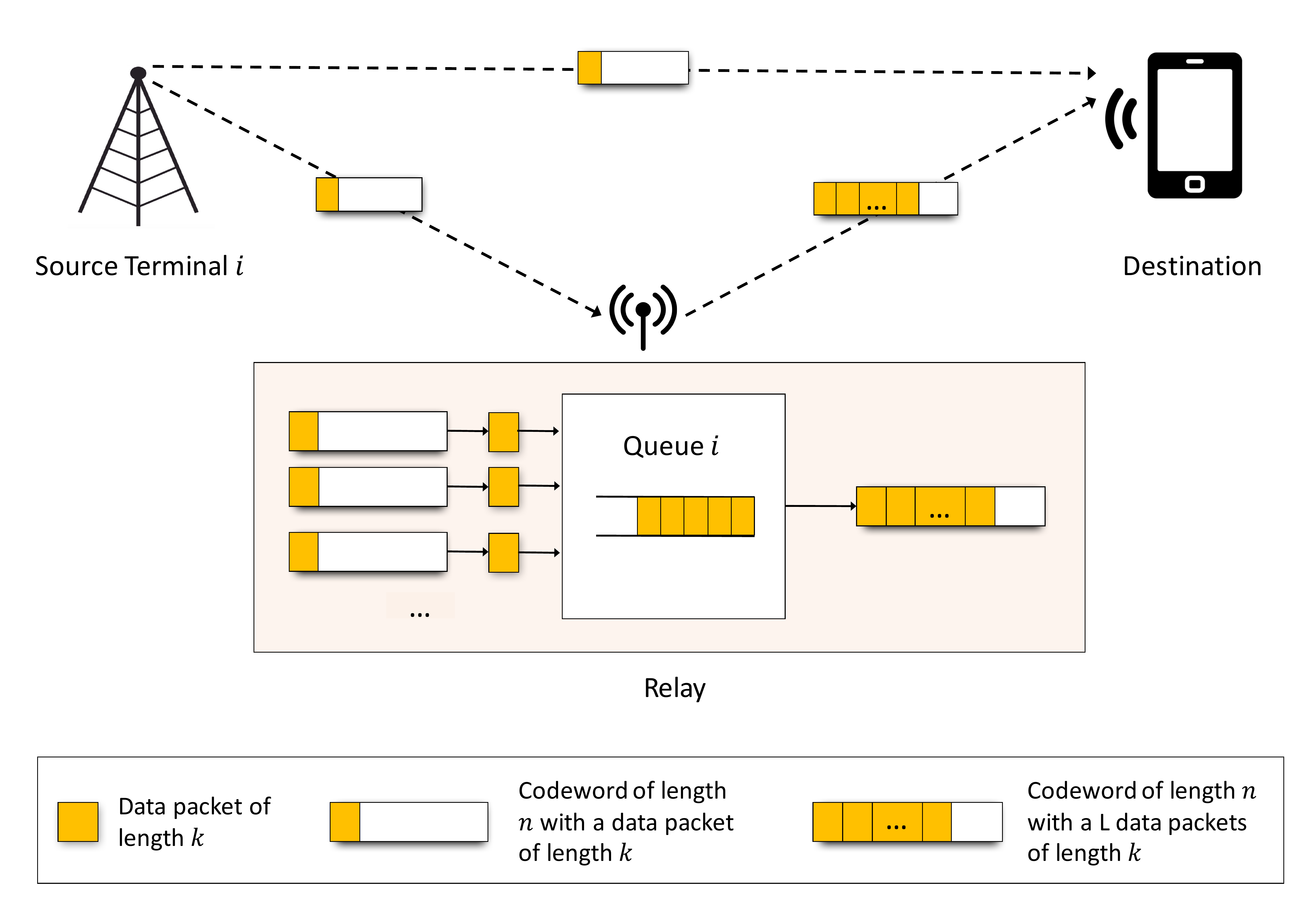}
  \caption{The Batch-And-Forward strategy for a source terminal $i$-relay pair.}
  \label{fig:marcbafck}
\end{figure} 
Motivated by the insufficient performance of the cognitive cooperation protocol, we propose a novel strategy that addresses the concerns encapsulated in Remark \ref{rem:cc}, and boosts the performance of the cognitive cooperation protocols in the finite blocklength regime. We evaluate the performance of the proposed strategy for the particular cognitive cooperation protocol given in Definition \ref{def:ccp}.

The proposed BAF strategy, depicted in Fig. \ref{fig:marcbafck}, keeps the data packet size the same for all individual nodes of the network, however, each node is allowed to batch more than one data packets into the codeword of length $n$. The number of data packets that are batched into the fixed length codeword, is denoted by $L$. Thus, this approach exploits the individual statistical characteristics of the different channels of the network, without imposing additional complexity on the overall scheme.

Next, we embed the BAF strategy at the relay of the cognitive cooperation protocol, and evaluate the performance of the overall network. This is implemented by replacing  step vi) of Definition \ref{def:ccp} with the following step.   
\addtolength{\itemindent}{0.2cm}
\begin{itemize}
\item[vi)] The relay batches $L$ data packets from the queue $Q_{iR}, i\in\{A,B\}$, and encodes them into a codeword of length $n$. When the source terminal $i\in\{A,B\}$ gains access to the channel, and it has no data packets in its queue (queue is idle),  the relay  transmits the codeword consisting of the $L$ data packets to the destination. If there are less than $L$ data packets in the respective queue at the relay, the relay does not transmit any information.
\end{itemize}
All the other procedures of Definition \ref{def:ccp} do not change.

\begin{theorem}\label{theo:baf}
Suppose  that the relay employs the BAF strategy and 
let $u^{BAF}(Lk,n)$ denote the overall throughput of cooperative scheme. Then, for a fixed blocklength $n$: 
\begin{itemize} 
\addtolength{\itemindent}{-0.25cm}
\item[i)]
The stability region of the BAF cooperative scheme satisfies
\begin{equation}
({\lambda_A}+ {\lambda_B})\frac{k}{n}< u^{BAF}(Lk,n), \label{eqthe0}
\end{equation}
where $k=1,2,\ldots$, $L=1,2,\ldots$, and
\begin{align}
u_S^{CC}(k,n)&=\frac{k}{n}\left[P_{c,SD}(k,n)+P_{e,SD}(k,n)P_{c,SR}(k,n)\right],\label{eqthe2}\\
u_R^{BAF}(Lk,n)&=\frac{Lk}{n}\dfrac{\left[P_{c,SD}(k,n)+P_{e,SD}(k,n)P_{c,SR}(k,n)\right]P_{c,RD}(Lk,n) }{\left[P_{c,RD}(Lk,n)+P_{e,SD}(k,n)P_{c,SR}(k,n)\right]},\label{eqthe3}\\
u^{BAF}(Lk,n)&=\min\left\{u_S^{CC}(k,n), u_R^{BAF}(Lk,n) \right\} \label{eqthe1}\\
&=\left\{
                \begin{array}{ll}
                  u_S^{CC}(k,n) \ \ \ \  \mbox{if} \ \ \ L\geq 1+\dfrac{P_{e,SD}(k,n)P_{c,SR}(k,n)}{P_{c,RD}(Lk,n)}, \\
                  u_R^{BAF}(Lk,n) \  \mbox{if} \ \ \ L< 1+\dfrac{P_{e,SD}(k,n)P_{c,SR}(k,n)}{P_{c,RD}(Lk,n)}.
                \end{array}
              \right. 
\end{align}
\item[ii)]  The code rate of the overall scheme, $X^{BAF}(k,n)$, is \addtolength{\itemindent}{-0.5cm} bounded above by the maximum throughput $u^{BAF,*}(k,n)$, that is
\beae
 X^{BAF}(k,n)<  u^{BAF,*}(k,n), \label{eq_ccsource10new2}
\eeae
where
\begin{equation}
u^{BAF,*}(Lk,n)=\max_{L, k}u^{BAF}(Lk,n).
\end{equation}
\end{itemize} 
\end{theorem}

\begin{proof} 
  See Appendix \ref{app_theo:baf}.
\end{proof}

\begin{figure*}%
\centering
\begin{subfigure}{1\columnwidth}
\center
\hspace{-0.5cm}
\includegraphics[width=0.8\columnwidth]{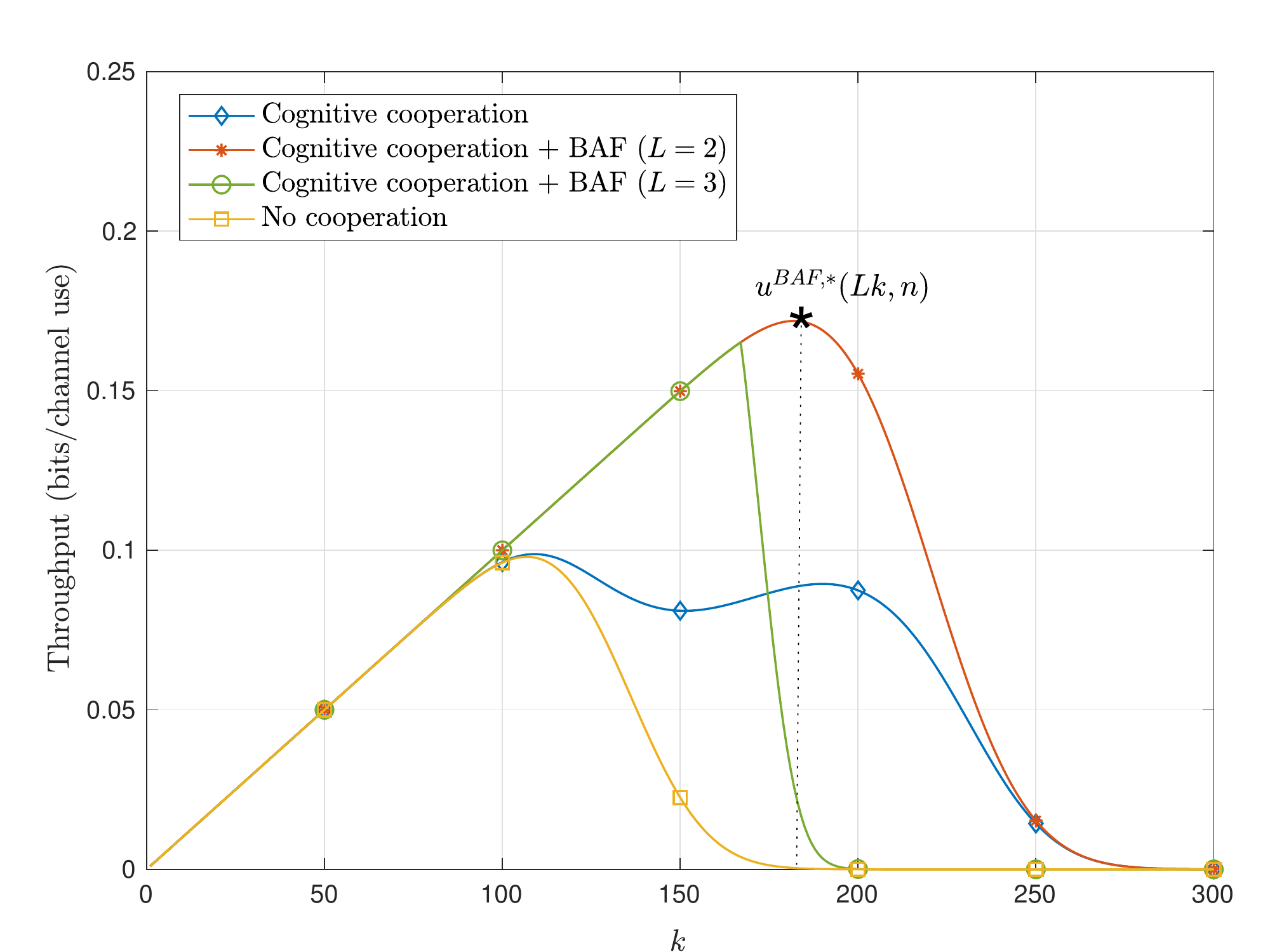}%
\vspace{-0.2cm}
\caption{}%
\label{subfiga1}%
\end{subfigure}\hfill%
\begin{subfigure}{1\columnwidth}\centering
\center
\includegraphics[width=0.8\columnwidth]{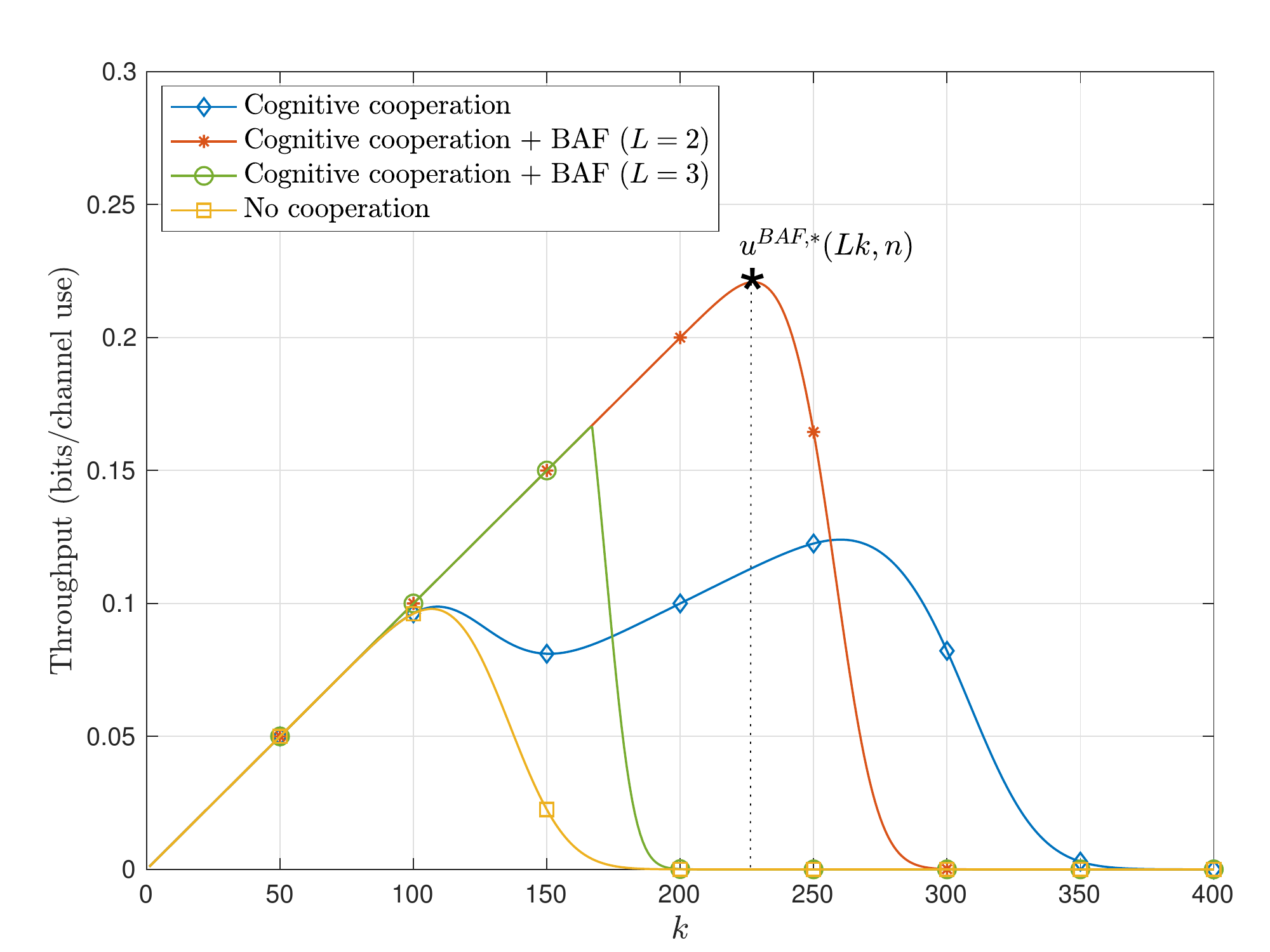}%
\vspace{-0.2cm}
\caption{}%
\label{subfigb1}%
\end{subfigure}\hfill%
\vspace{-0.2cm}
\caption{Throughput of the cognitive cooperation protocol embedded with the BAF strategy, for fixed blocklength $n=1000$, and comparison with the no cooperation scheme and the cognitive cooperation protocol. The channels from the source to the destination, from the source to the relay, and from the relay to the destination, are AWGN with $SNR$ for (a) 0.2, 0.35 and 1, respectively, and for (b) 0.2, 0.5 and 1, respectively.}
\label{figabc1}
\end{figure*}

The performance of the cognitive cooperation protocol with BAF strategy at the relay, is illustrated in Fig. \ref{figabc1}.
For the selected $SNR$ triplet depicted in Fig. \ref{figabc1}(a),
the optimal data packet size is 182 bits, whereas for the selected $SNR$ triplet depicted in Fig. \ref{figabc1}(b),
the optimal data packet size is 227 bits. As it is depicted, the optimal batching size,  for both cases, is $L=2$. It is also obvious that for both cases, the BAF strategy can significantly enhance the performance of the overall system (approximately by 75\%, in both scenarios), compared to the cognitive cooperation protocol without BAF. For both of the scenarios above, the $SNR$ of the channel between the relay and the destination, is higher than the $SNR$ of the channel between the source terminal and the destination, thus, is beneficiary for the overall performance of the network to apply the BAF strategy at the relay. For other scenarios in which the  $SNR$ of the channel between the source terminal and the destination, is higher than the $SNR$ of the channel between the relay and the destination, it would have been beneficial to have applied the BAF strategy at the source terminals.

Next we provide the overall stability conditions and the optimal throughput, when the BAF strategy is employed at the source terminals instead of the relay.

\begin{corollary}\label{cor:baf}
Suppose  that the BAF strategy is employed by the source terminal and 
let $u^{SBAF}(Lk,n)$ denote the overall throughput of cooperative scheme. Then, for a fixed blocklength $n$: 
\begin{itemize} 
\addtolength{\itemindent}{-0.25cm}
\item[i)]
The stability region of the BAF cooperative scheme satisfies
\begin{equation}
({\lambda_A}+ {\lambda_B})\frac{k}{n}< u^{SBAF}(Lk,n), \label{eqthe0s}
\end{equation}
where $k=1,2,\ldots$, $L=1,2,\ldots$, and
\begin{align}
u_S^{BAF}(Lk,n)&=\frac{Lk}{n}\left[P_{c,SD}(Lk,n)+P_{e,SD}(Lk,n)P_{c,SR}(Lk,n)\right],\label{eqthe2s}\\
u_R^{CC}(k,n)&=\frac{k}{n}\dfrac{\left[P_{c,SD}(Lk,n)+P_{e,SD}(Lk,n)P_{c,SR}(Lk,n)\right]P_{c,RD}(k,n) }{\left[P_{c,RD}(k,n)+P_{e,SD}(Lk,n)P_{c,SR}(Lk,n)\right]},\label{eqthe3s}\\
u^{SBAF}(Lk,n)&=\min\left\{u_S^{BAF}(Lk,n), u_s^{BAF}(k,n) \right\} \label{eqthe1s}\\
&=\left\{
                \begin{array}{ll}
                  u_S^{BAF}(Lk,n) \ \ \ \  \mbox{if} \ \ \ L\geq \dfrac{P_{c,RD}(Lk,n)}{\left[P_{c,RD}(k,n)+P_{e,SD}(Lk,n)P_{c,SR}(Lk,n)\right]}, \\
                  u_R^{CC}(Lk,n) \ \ \ \ \   \mbox{if} \ \  \ L< \dfrac{P_{c,RD}(Lk,n)}{\left[P_{c,RD}(k,n)+P_{e,SD}(Lk,n)P_{c,SR}(Lk,n)\right]}.
                \end{array}
              \right. 
\end{align}
\item[ii)]  The code rate of the overall scheme, $X^{SBAF}(k,n)$, is \addtolength{\itemindent}{-0.5cm} bounded above by the maximum throughput $u^{SBAF,*}(k,n)$, that is $X^{SBAF}(k,n)<  u^{SBAF,*}(k,n)$,
where
\begin{equation}
u^{SBAF,*}(Lk,n)=\max_{L, k}u^{SBAF}(Lk,n).
\end{equation}
\end{itemize} 
\end{corollary}
\begin{proof} 
The proof is similar to the proof of Theorem~\ref{theo:baf}, thus is omitted.
\end{proof}

Note that that the optimal performance of the network is then obtained by maximizing the outcomes of Theorem~\ref{theo:baf} and Corollary~\ref{cor:baf}. That is, the optimal throughput of the network, $u^{*}$, is given by
\bea
u^{*}=\max\left\{u^{SBAF,*},u^{BAF,*}\right\}.
\eea

Though the majority of the classical cooperative techniques can be employed for short packet communication, they cannot fully correspond to the special characteristics of short codes, since, they were not particularly designed to perform optimally in the finite blocklength regime. The proposed approach, however, can significantly enhance the performance  of the network, while at the same time it meets the finite blocklength requirements. Moreover, it can reduce the requirements in metadata, a challenging task in the actual implementation of short codes \cite{popo2016}, since, it avoids the unnecessary repetition of metadata (e.g. address of the source terminal and the destination). Perhaps the most attractive feature, however, of the BAF strategy, is that it can be embedded into existing cooperative protocols, without imposing any additional complexity to the system. 

\section{Conclusion}

In this work, we employed tools and results from information theory, stochastic processes and queueing theory, in order to provide a comprehensive framework regarding the analysis of a TDMA network with bursty traffic, in the finite blocklength regime. In particular,  we examined the stability of a TDMA network,  evaluated the optimal throughput for fixed blocklength constraints, and identified the optimal trade-off between data length and latency, both numerically and via the proposed closed form approximations. Moreover, we examined the MARC-TDMA network, evaluated the stability conditions for a particular cognitive cooperation protocol, and proposed the BAF strategy that can enhance the finite blocklength performance of cognitive protocols. The BAF strategy can be easily embedded in existing cooperative techniques without imposing additional complexity. In the current work we did not address issues regarding metadata, such as, impact of metadata on the performance and design of metadata for short codes. This is a challenging task for the performance analysis of finite blocklength analysis, that will be investigated as a part of future work.

\begin{appendix}
\label{appendix}
\subsection{Proof of Theorem~\ref{theo:main1}.}
\label{appendix:theom1}
The stability conditions of the underlying Markov chains at the two terminals depend on the existence, or non-existence, of a stationary distribution, defined by
\begin{equation}
{\pi}_{i,j}=\lim_{m\rar\infty}P(S_m=j), \ j\geq 0, i\in\{A,B\}.
\end{equation}
The characterization of the stationary distribution for the emerged $Geo/Geo/1$ queue is obtained by employing the global balance equations \cite{robertazzi2007}, which yield 
\beae
{\pi}_{i,0}&=&\dfrac{1-q_i}{-q_1+\mathlarger{\mathlarger{‎‎\sum}}_{m=0}^{\infty}\left(\dfrac{p_i(1-q_i)}{q_i(1-p_i)}\right)^m}=\dfrac{q_i-p_i}{q_i}, \ i\in\{A,B\},\nms \label{eq:stationary_dis0}\\
{\pi}_{i,j}&=&\left(\dfrac{p_i(1-q_i)}{q_i(1-p_i)}\right)^{j}\dfrac{1}{1-q_i}\pi_{i,0}, \ j\geq 1, \ i\in\{A,B\}.
\eeae
Therefore, the stationary distribution is non-zero, only if 
\begin{equation}
\dfrac{p_i(1-q_i)}{q_i(1-p_i)}<1,  \ \ \forall i\in\{A,B\},
\end{equation}
or equivalently
\begin{equation}
 q_i>p_i,  \ \ \forall i\in\{A,B\}. \label{eq:stability_condition1}
\end{equation}
Otherwise, ${{{‎‎\sum}}_{i=0}^{\infty}({p_i(1-q_i)}/{q_i(1-p_i)})}$ would be infinite and all $\{{\pi}_{i,j}, \ j\geq0\}$ would be zero. By substituting the average arrival rate, $\lambda_i=p_i$, and average departure rate $\mu_i=q_i=\omega_i(1-P_{e,i}(k_i,n))$, in \eqref{eq:stability_condition1}, we obtain the following stability condition
\begin{equation}
 {\lambda_i}<{\omega_i(1-P_{e,i}(k_i,n))}\tri\omega_i P_{c,i}(k,n),  \ \ \forall i\in\{A,B\}. \label{eq:stability_condition2}
\end{equation}
Summing over all $i\in\{A,B\}$, we obtain \eqref{eq:throughputeqd0}, which completes the proof.

\subsection{Proof of Proposition~\ref{prop:lin}.}
\label{app_prop:lin}

 The parameters $\{\delta_0,\delta_1\}\in\mathbb{R}$, are evaluated by minimizing  the integral of the absolute error, that is
\bea
\big\{\delta_0^*,\delta_1^* \big\}=\argmin_{\delta_0,\delta_1}\int_{-\infty}^{\infty}\Big{|}  {\hat{P}_c}(k,n)-P_c(k,n)\Big{|}
d{\chi},\eea
which results to $\delta_0=0.5$ and $\delta_1=1.545$. Then, by employing the above approximation, the optimization problem is given by
\bea
{\hat u}(k,n)=\max_{k}\dfrac{k}{n}{\hat P}_c(k,n)=\max_{k}
  \begin{cases}
    \dfrac{k}{n} & \text{if $\chi \geq \delta_1$}, \\
    \dfrac{k}{n}\left(\dfrac{1}{2\delta_1}\chi+\delta_0\right)& \text{if $-\delta_1 \leq \chi < \delta_1$}. \label{eq:linappopt}
  \end{cases} 
\eea 
We first perform the optimization in the region $-\delta_1 \leq \chi < \delta_1$.   By substituting $\chi$ and the value of $\delta_1$, we rewrite  the predefined region as a function of $k$, that is
\begin{equation}
nC-1.545\sqrt{nV}< k \leq nC+1.545\sqrt{nV}.\label{eq:optreg1lin}
\end{equation}
The optimization problem is solved by differentiating the objective function, ${\hat u}(k,n)$, with respect to $k$, and verifying that the second derivative is negative. The optimal value of $k$ is then given by
\begin{equation}
{\hat k}^{*}=0.5\left(Cn+1.545\sqrt{nV}\right). \label{eq:optsollin}
\end{equation}
The analytical calculations are omitted due to space limitations. Since the value of $k$ must lay in the region defined by \eqref{eq:optreg1lin},  the optimal value of $k$ is valid only if 
\begin{equation}
nC-1.545\sqrt{nV}< {\hat k}^{*} \leq nC+1.545\sqrt{nV}.\label{eq:optreg1lin1}
\end{equation}
By substituting \eqref{eq:optsollin} in  \eqref{eq:optreg1lin1} and solving with respect to the blocklength $n$, we obtain the region of $n$ for which the optimal solution given by \eqref{eq:optsollin} holds, which yields $0\leq n<{13.905V}/{C^2}$. For the region $\chi\geq\delta_1$, or equivalently for $k\leq nC-1.545\sqrt{nV}$, the maximization of $k/n$ with respect to $k$, occurs on the boundary, that is,  $k= nC-1.545\sqrt{nV}$, and this solution holds for $n\geq{13.905V}/{C^2}$. The optimal throughput is then obtained by substituting the optimal size of the data packet, ${\hat k}^{*}$, in \eqref{optim_pro1}.

\subsection{Proof of Proposition~\ref{prop:quad}.}
\label{app_prop:quad}

Since, (i) the approximation given by \eqref{eq:sec_ord_app} is odd-symmetric with respect to $\chi=0$, and (ii) $P_c(k,n)\Big{|}_{\chi=0}=0.5$, then the optimal value of $\theta_0$ that minimizes the absolute value of the error between ${\tilde{P}_c}(k,n) $ and ${{P}_c}(k,n)$ is, $\theta_0=0.5$.

Next, we evaluate the parameters $\{\theta_1,\theta_2\}$, by imposing an additional constraint regarding the continuity of the first derivative with respect to $k$, which significantly  simplifies the optimization problem. The proposed quadratic form guarantees continuity in the region $\chi\in(-\theta_1,\theta_1)$. The conditions that ensure continuity of the first derivative in the regions $\chi\in(-\infty,-\theta_1]$ and  $\chi\in[\theta_1,\infty)$, and thus for the whole region $\chi\in(-\infty,\infty)$, are
\beae
&\dfrac{d}{d{k}}\left[{\tilde{P}_c(k,n)}\right]\Big{|}_{\chi=\theta_1}=0,\label{eqc1new}\\ &\dfrac{d}{d{k}}\left[{\tilde{P}_c(k,n)}\right]\Big{|}_{\chi=-\theta_1}=0,\label{eqc2new}\\
&{\tilde{P}_c(k,n)\Big{|}_{\chi=\theta_1}=1},\label{eqc3new}\\
&{\tilde{P}_c(k,n)\Big{|}_{\chi=-\theta_1}=0}.\label{eqc4new}
\eeae
Equations \eqref{eqc1new} and \eqref{eqc2new} are satisfied directly by the proposed quadratic form, whereas equations \eqref{eqc3new} and \eqref{eqc4new} are satisfied, if and only if, $\theta_2=0.5/{\theta_1}^2$. The remaining parameter, $\theta_1$, is evaluated by minimizing  the integral of the absolute error 
\bea
\big\{\theta_1^* \big\}=\argmin_{\theta_1}\int_{-\infty}^{\infty}\Big{|}  {\tilde{P}_c}(k,n)-P_c(k,n)\Big{|}d{\chi},
\eea
which yields $\theta_1=2.35$.
The optimization problem for the case of the quadratic approximation is 
\bea
\tilde{u}(k,n)=\max_{k}\dfrac{{k}}{n}{\tilde{P}}_c(k,n),\label{eq:quadappopt}
\eea 
where ${\tilde{P}}_c(k,n)$ is given by \eqref{eq:sec_ord_app}. By employing the methodology discussed in Appendix \ref{app_prop:lin}, we obtain the optimal data packet size given by \eqref{eq:finaloptsolquad}. Then, the optimal throughput, ${\tilde u}^{*}(k,n)$, emerges by substituting the values of \eqref{eq:sec_ord_app} and \eqref{eq:finaloptsolquad}, in $\dfrac{\tilde{k}^{*}}{n}\tilde{P}^{*}_c(k,n)$.

\subsection{Proof of Theorem~\ref{theo:coopmarch}.}
\label{app_theo:coopmarch}
i) The proof for the first statement consists of two parts: characterizing the stability regions for the source terminals and the relay, and evaluating  the union of the predefined regions. Regarding the source terminals, the arrival rate at each source terminal is $\lambda_i, \  i \in\{A,B\}$, whereas the departure (service) rate is given by
\begin{equation}
{\mu_i}=\omega_i\left[P_{c,SD}(k,n)+P_{e,SD}(k,n)P_{c,SR}(k,n)\right], \ \ i\in\{A,B\}. \label{eq_ccsource1}
\end{equation}
The system at each source terminal $i$ forms a discrete-time Markov chain with stability condition $\frac{\lambda_i}{\mu_i}<1, \ \forall i\in\{A,B\}$ \cite{alfa}, or equivalently
\begin{equation}
{\lambda_i}<\omega_i\left[P_{c,SD}(k,n)+P_{e,SD}(k,n)P_{c,SR}(k,n)\right], \ \forall i\in\{A,B\}. \label{eq_ccsource2}
\end{equation}
By recalling that $\omega_A+\omega_B=1$, the summation  \eqref{eq_ccsource2} over all source terminals $i\in\{A,B\}$ yields the following overall stability condition for the source terminals.
\beae
{\Lambda}^{CC}_{S}=\Big\{(\lambda_A, \lambda_B): 
  ({\lambda_A}+ {\lambda_B}) <\left[P_{c,SD}(k,n)+P_{e,SD}(k,n)P_{c,SR}(k,n)\right]
\Big\}. \label{eq_ccsource3}
\eeae
Regarding the relay, a packet from the source terminal $i, \ i\in\{A,B\}$, enters queue $Q_{iR}$ at the relay if i) channel access for the source terminal $i$ is granted by the 
randomized switch (with probability $\omega_i$),  ii) the transmission from the source terminal $i$ to the  relay is successful, iii) the transmission from the source terminal $i$ to the destination is unsuccessful, and iv) the queue of the source terminal $i$ is not idle, that is, it has at least one packet that requires transmission. The source is not idle with stationary probability $1-\pi_{i,0}$, where $\pi_{i,0}=1-\frac{\lambda_i}{\mu_i}$ \cite{alfa}. Thus, the rate of arrivals at relay's queue $Q_{iR}, \ i\in\{A,B\}$, is given by
\begin{equation}
\lambda_{iR}=\omega_i(1-\pi_{i,0})P_{e,SD}(k,n) P_{c,SR}(k,n), \ \ i\in\{A,B\}.  \label{eq_ccsource5}
\end{equation}
Similarly, the average rate of departures from the relay's queue $Q_{iR}, \ i\in\{A,B\}$, is given by
\begin{equation}
\mu_{iR}=\omega_i\pi_{i,0}P_{c,RD}(k,n), \ \ i\in\{A,B\}.  \label{eq_ccsource6}
\end{equation}
The stability condition for the individual queue $Q_{iR}, \ i\in\{A,B\}$ at relay, is $\frac{\lambda_{iR}}{\mu_{iR}}<1, \  \forall i\in\{A,B\}$. By employing \eqref{eq_ccsource1}, \eqref{eq_ccsource5} and \eqref{eq_ccsource6}, the stability condition, $\forall i\in\{A,B\}$, translates to
\beae
\lambda_{iR}<\frac{\omega_{i}\left[P_{c,SD}(k,n)+P_{e,SD}(k,n)P_{c,SR}(k,n)\right]P_{c,RD}(k,n) }{\left[P_{c,RD}(k,n)+P_{e,SD}(k,n)P_{c,SR}(k,n)\right]}, \nms \label{eq_ccsource7}
\eeae
whereas the overall stability region, obtained by summing \eqref{eq_ccsource7} over all $i\in\{A,B\}$, is given by 
\beae
{\Lambda}^{CC}_{R}=\Bigg\{(\lambda_A, \lambda_B): ({\lambda_A}+ {\lambda_B})<\dfrac{\left[P_{c,SD}(k,n)+P_{e,SD}(k,n)P_{c,SR}(k,n)\right]P_{c,RD}(k,n) }{\left[P_{c,RD}(k,n)+P_{e,SD}(k,n)P_{c,SR}(k,n)\right]}\Bigg\}.\nms
\label{eq_ccsource12}
\eeae
The overall stability region is given by the union of \eqref{eq_ccsource3} and \eqref{eq_ccsource12}. However, by comparing \eqref{eq_ccsource3} and \eqref{eq_ccsource12} we observe that ${\Lambda}^{CC}_{R} \subseteq	{\Lambda}^{CC}_{S}$, thus, the overall stability region ${\Lambda}^{CC}={\Lambda}^{CC}_{R}$.

ii) The second statement is obtained by  maximizing \eqref{eq_ccsource10} with respect to $k$.

\subsection{Proof of Theorem~\ref{theo:baf}.}
\label{app_theo:baf}
%

i)  The stability region for the source terminals is identical to the cognitive cooperation case, thus is  given by \eqref{eqthe2}. Similarly, the arrival rate at the relay terminal, is also identical with the  arrival rate of the cognitive cooperation scheme, and is given by \eqref{eq_ccsource5}.

Since $L$ packets are batched together, the total size of the data packets that are encoded  into a codeword of length $n$, is $Lk$. The departures from the  queue $Q_{iR}, i\in\{A,B\}$ at the relay, are also Bernoulli distributed, with departure probability, at a given time slot, $q_{iR}=\omega_i(\pi_{i,0})P_{c,RD}(Lk,n), \ i\in\{A,B\}$. The average departure rate is therefore given by 
\begin{eqnarray}
 \mu_{iR}={L}q_{iR}={L}\omega_i(\pi_{i,0})P_{c,RD}(Lk,n), \ i\in\{A,B\}. \label{appeq3}
\end{eqnarray}
The appropriate model at each queue of the relay is the $Geo/Geo^{L}/1$ model \cite{alfa}. In this model both the arrivals and the departures to and from the relay are Bernoulli distributed, and packets depart in batches of size $L$. 

The stability condition for $Geo/Geo^{L}/1$ queue \cite{alfa}, is given by 
\begin{eqnarray}
\frac{\lambda_{iR}}{ \mu_{iR}}<L, \ \ \forall i\in\{A,B\}. \label{appeq4}
 \end{eqnarray}
By substituting  \eqref{eq_ccsource5} and \eqref{appeq3} in \eqref{appeq4}, and by summing over $i\in\{A,B\}$, we obtain \eqref{eqthe3}.
Then, the stability of the overall network is given by the union of  \eqref{eqthe2} and \eqref{eqthe3}, which yields \eqref{eqthe1}.

ii) The maximum throughput of the system is obtained by maximizing \eqref{eqthe1} with respect to $k$ and $L$.
\end{appendix}

\bibliographystyle{IEEEtran}
\bibliography{Bibliography}

\end{document}